\documentclass[11pt]{article}

\usepackage[margin=1in]{geometry}
\usepackage[T1]{fontenc}
\usepackage[utf8]{inputenc}
\usepackage{lmodern}
\usepackage{microtype}
\usepackage{amsmath, amssymb, amsthm}
\usepackage{booktabs, tabularx, longtable, array}
\usepackage{enumitem}
\usepackage{xcolor}
\usepackage[section]{placeins}
\usepackage{hyperref}
\usepackage[numbers,sort&compress]{natbib}
\usepackage{graphicx}
\usepackage{tikz}
\usepackage{pgfplots}
\usetikzlibrary{arrows.meta, fit, positioning, shapes.geometric}
\pgfplotsset{compat=1.18}

\hypersetup{
  colorlinks=true,
  linkcolor=blue!50!black,
  citecolor=blue!50!black,
  urlcolor=blue!50!black,
  pdftitle={Proof-Carrying Agent Actions},
  pdfauthor={Zexun Wang},
}

\setlist[itemize]{leftmargin=1.4em}
\setlist[enumerate]{leftmargin=1.6em}

\newtheorem{definition}{Definition}
\newtheorem{proposition}{Proposition}

\definecolor{pcaastatic}{HTML}{B6B6B6}
\definecolor{pcaascalar}{HTML}{F28C38}
\definecolor{pcaaruntime}{HTML}{2E7DB6}
\definecolor{pcaarisk}{HTML}{D9534F}
\definecolor{pcaawarn}{HTML}{E8A845}

\pgfplotsset{
  pcaa chart/.style={
    width=0.9\linewidth,
    height=0.34\textheight,
    ymajorgrids=true,
    grid style={dashed, gray!25},
    axis line style={black!45},
    tick style={black!45},
    xlabel style={font=\small},
    ylabel style={font=\small},
    tick label style={font=\scriptsize},
    title style={font=\small},
    legend style={font=\scriptsize, draw=none, fill=none},
  },
}

\title{
  Proof-Carrying Agent Actions:\\
  Model-Agnostic Runtime Governance for\\
  Heterogeneous Agent Systems
}
\author{Zexun Wang\thanks{Correspondence: \texttt{jw@nd.im}}\\Ond Holdings Inc.}
\date{May 31, 2026}

\begin{document}

\maketitle

\begin{abstract}
Agent systems now execute through runtimes with very different control points: local coding tools, framework SDKs, managed agent platforms, API gateways, and observer-only integrations. A high-risk action such as publishing data externally may therefore appear as a shell command in one runtime, a tool call in another, and a hosted session transition in a third. This heterogeneity makes it difficult to answer a basic governance question consistently: what action was authorized, under whose authority, with what approval semantics, and with what evidence after execution?

This paper presents \emph{Proof-Carrying Agent Actions} (PCAA), a runtime-neutral governance model centered on an action certificate rather than on a vendor-native session record. PCAA organizes control around five checkpoints---pre-action admissibility, action open, assumption capture, approval, and outcome closure---and binds them to a portable action envelope, runtime and approval receipts, and replay-ready proof. The model is extended in two practical ways that became necessary in deployment: the certificate is \emph{externality-aware}, carrying boundary facts such as destination visibility and account provenance, and approval is described by explicit enforceability classes rather than by a single reviewed/unreviewed bit.

We study the model through a reference implementation in a heterogeneous agent control plane and a disclosure-bounded evaluation protocol. On a protected benchmark expanded from 24 executable seeds to 96 traces across four runtime families, PCAA preserves route quality while exposing different failure modes under ablation: removing externality context degrades routing, removing approval-enforceability handling shifts review posture, and removing the integrity lane collapses proof stability without changing routing itself. The paper's contribution is therefore twofold: a systems formulation of runtime governance around certificate-bearing actions, and an implementation-grounded account of how that formulation can remain portable under runtime churn without collapsing into vendor-specific control surfaces.
\end{abstract}

\section{Introduction}

The same high-impact agent action can surface in strikingly different forms. Consider an outbound publication step: in one environment it appears as a shell command executed by a coding agent, in another as an SDK tool call, and in a third as a hosted session transition behind an API gateway. The operational question, however, does not change with the runtime. Someone still has to decide whether the action may proceed, whether it requires human review, what counts as acceptable evidence, and how the resulting decision can later be audited or challenged.

Existing systems usually make one neighboring object legible. System cards describe model-family risks. Managed runtimes expose sessions and tool traces. Guardrail frameworks expose validators and rails. Cryptographic systems expose commitments and attestations. These artifacts are useful, but they answer different questions. In particular, they do not by themselves provide a stable trust object for heterogeneous agent control: one record that says what action was admitted, under which authority, with what approval semantics, and with what replayable evidence.

This paper studies \emph{Proof-Carrying Agent Actions} (PCAA) as one answer to that problem. The central design choice is to treat the \emph{action certificate} as the primary trust-bearing object. In the reference implementation studied here, that certificate path is operationalized as a route-review-prove method:

\begin{enumerate}
  \item route the action through a stable governance vocabulary,
  \item review the action when policy or ambiguity requires human oversight,
  \item prove the result in a form that survives replay, export, and runtime change.
\end{enumerate}

The current formulation sharpens that idea into an explicit systems contract. It defines a five-checkpoint sequence, a portable action envelope, a neutral runtime contract, approval and runtime receipts, proof bundles, and a layered integrity architecture. Just as importantly, it makes the authority split explicit. External runtimes or partner governors may enforce policy close to execution, but the workspace-visible certificate remains responsible for admission, approval, replay, and proof closure.

The paper makes five contributions:

\begin{enumerate}
  \item a runtime-neutral definition of PCAA centered on certificate-bearing actions rather than runtime replacement;
  \item a formal five-checkpoint contract for governing high-value actions;
  \item a portable action-envelope model that carries authorization, accountability, workflow, evaluation, and externality context together;
  \item a runtime contract that distinguishes families, adapter modes, authority boundaries, and approval-enforceability classes while remaining honest about heterogeneous control depth;
  \item an implementation-grounded proof model that joins replay bundles, runtime-stage receipts, approval receipts, and integrity projections on one certificate path.
\end{enumerate}

The empirical question is narrower than the architectural one. We do not claim universal pre-execution control, nor do we claim that every runtime can expose the same receipt depth. Instead, we ask whether a certificate-centered control path can remain portable across changing runtime families without losing review semantics or proof closure. The answer offered here is implementation-grounded rather than purely theoretical: a heterogeneous control plane can preserve a common governance object if the runtime-specific details are projected into an explicit certificate path rather than left inside runtime-native tooling.

\section{Related Work}

\subsection{Auditable agents and accountability}

Recent work on auditable agents argues that action recoverability, lifecycle coverage, policy checkability, responsibility attribution, and evidence integrity are prerequisites for meaningful accountability \citep{nian2026auditable}. PCAA aligns closely with this direction, but narrows the systems primitive further: the trust-bearing unit should be the action certificate produced at decision time, not only the audit surface assembled after the fact.

\subsection{Managed runtimes, observability, and guardrail systems}

Current frontier systems optimize adjacent but different objects. OpenAI system cards report model-family risk and deployment mitigations \citep{openai2023gpt4, openai2026gpt5systemcard}. Anthropic Managed Agents and Azure AI Foundry Agent Service emphasize the hosted runtime harness, including sessions, tools, execution state, and enterprise hosting controls \citep{anthropic2026managedannouncement, anthropic2026managedoverview, microsoft2026foundry}. LangSmith centers traces and evaluation workflows \citep{langchain2026langsmith}. NeMo Guardrails and Guardrails AI center policy rails, validators, and structured remediation \citep{nvidia2026nemoguardrails, guardrails2025docs}. Lakera Guard centers prompt and content threat decisions \citep{lakera2026guard}. PCAA is intended to compose with these systems, but differs in what it treats as primary: action-level admissibility, replay-ready evidence, and certificate closure.

\subsection{Selective control and abstention-aware risk}

Conformal risk control and selective conformal variants show that abstention-aware decision systems can enforce calibrated risk constraints under selective deployment \citep{angelopoulos2022crc, zhou2025scrc}. PCAA does not currently implement a full conformal pipeline in the product, but its routing semantics are aligned with the same intuition: high-cost review and escalation should be allocated selectively rather than uniformly.

\subsection{Proof-carrying and verifiable artifacts}

Proof-carrying code established the general idea that an executable artifact can travel with a machine-checkable proof object \citep{necula1997pcc}. PCAA adapts that spirit to runtime action governance. The target is not low-level code safety, but action admissibility under changing runtimes. Modern verifiable credential and attestation ecosystems also matter here. Typed portable claims and attestations are foundational for moving selected trust material across organizational boundaries \citep{w3c2025vc, eas2025docs}. Likewise, zkVM receipts and finalized blockchain settlement states motivate stronger future integrity layers without requiring that every governed action already be fully externally verified \citep{risczero2026zkvm, solana2026production, solana2026verification}.

\subsection{What is actually new}

The novelty claim is intentionally architectural rather than rhetorical. PCAA does not claim to invent model safety evaluation, managed agent execution, runtime policy enforcement, traces, programmable rails, or cryptographic attestation in isolation. The proposed novelty is their recomposition around a different systems primitive: an action-level certificate that is portable across runtimes, bound to selective review semantics, replayable from stored evidence, and attachable to integrity or attestation lanes without surrendering final authority.

\section{Preliminaries and Notation}

Let $\mathcal{A}$ denote the space of raw runtime actions. A raw action may come from a local coding hook, an SDK callback, a managed runtime session, a browser operation, an API request, or a bridge translating an external governor into OSuite vocabulary.

The implementation studied here is best modeled through a portable action envelope rather than a single opaque score:

\begin{definition}[PCAA Certificate]
For a raw action $a \in \mathcal{A}$, let
\[
\mathcal{C}(a) = \bigl(E(a), Q_r(a), Q_h(a), P(a), I(a)\bigr),
\]
where:
\begin{itemize}
  \item $E(a)$ is the normalized action envelope,
  \item $Q_r(a)$ is the set of runtime-stage receipts,
  \item $Q_h(a)$ is the set of approval receipts and accountability facts,
  \item $P(a)$ is the replay and proof state attached to the action,
  \item $I(a)$ is the integrity projection, including commitments, attestations, or future verification references.
\end{itemize}
\end{definition}

This formulation emphasizes canonicalization, risk, uncertainty, admissibility, proof-graph summary, and commitment material through the action-envelope and proof-bundle contracts rather than treating them as disconnected abstract objects.

\begin{table}[t]
\centering
\caption{Core notation used in the paper.}
\begin{tabularx}{\linewidth}{@{}>{\raggedright\arraybackslash}p{0.18\linewidth}X@{}}
\toprule
Symbol & Meaning \\
\midrule
$a \in \mathcal{A}$ & Raw runtime action emitted by a runtime, bridge, or adapter \\
$E(a)$ & Portable action envelope attached to the action \\
$Q_r(a)$ & Runtime-stage receipts such as pre-input or pre-tool decisions \\
$Q_h(a)$ & Human approval and accountability receipts \\
$P(a)$ & Replay-ready proof state, including summaries and bundle projections \\
$I(a)$ & Integrity references such as commitments, attestations, or future receipts \\
$F(a)$ & Final authority designation, expected to resolve to PCAA at closure \\
$B_t$ & Available review or audit budget at time $t$ \\
$\pi(a)$ & Admissibility route selected for the action \\
\bottomrule
\end{tabularx}
\end{table}

In the reference implementation, the action envelope contains at least the following classes of fields:

\begin{itemize}
  \item identity and tenant scope,
  \item runtime-family and adapter-mode declarations,
  \item authorization context,
  \item workflow context,
  \item evaluation context,
  \item accountability context,
  \item optional runtime, partner, and integrity projections.
\end{itemize}

The key shift is that PCAA is not expressed as a hidden policy function. It is expressed as a stable runtime governance object whose fields are inspectable, replayable, and exportable.

\section{Problem Setting}

\subsection{Action governance under runtime heterogeneity}

Depending on the runtime family, a raw action may appear as:

\begin{itemize}
  \item a tool-use invocation in a coding agent,
  \item an SDK method call inside a framework runtime,
  \item a workflow transition in a managed platform,
  \item a gateway request crossing an API boundary,
  \item a browser-native operation,
  \item a signed request, wallet intent, or protocol-lane message.
\end{itemize}

The operator-facing problem is therefore not simply whether an action is ``safe.'' The actual control problem is:

\begin{quote}
Given a proposed action $a$, determine whether it should be allowed immediately, routed through additional evidence checks, held for approval, denied, deferred under uncertainty, or recorded as observe-only under partial governance coverage.
\end{quote}

\subsection{Desired properties}

The model studied here requires the following non-negotiable properties:

\begin{enumerate}
  \item \textbf{Runtime neutrality.} Governance cannot collapse into one runtime vendor or session schema.
  \item \textbf{Explicit authority closure.} Final authority must remain visible and tenant-readable.
  \item \textbf{Portable review semantics.} Approval cannot disappear into runtime-specific product states.
  \item \textbf{Replay and exportability.} The resulting trust object must survive replay, proof export, and buyer-facing review.
  \item \textbf{Honest coverage disclosure.} The system must distinguish strong interception from bridge or observe-only modes.
  \item \textbf{Extensible integrity.} Fast commitments, portable attestations, and future verification lanes must compose without redefining the base control object.
\end{enumerate}

\subsection{Research objective}

Abstractly, PCAA seeks a control mechanism $M$ that minimizes severe missed escalations while respecting review budgets and runtime-boundary constraints:
\[
\min_M \; \mathbb{E}[\ell_{\mathrm{miss}}(a, M)]
\quad \text{s.t.} \quad
\mathbb{E}[c(a)\,\mathbf{1}\{\mathrm{review}(a)\}] \leq B_t.
\]

The important point is not the specific loss function. It is the shift in objective. PCAA does not treat average scoring quality as the sole endpoint. It treats selective routing, explicit review semantics, and certificate closure as the actual runtime trust problem.

\section{Proof-Carrying Agent Actions}

\subsection{Boundary}

PCAA is not a replacement runtime. It is the authority layer that closes admission, approval, and proof. This boundary is central to the formulation studied here. External runtimes, managed execution layers, or partner systems may contribute receipts, trust signals, or enforcement close to execution, but they do not replace the tenant-visible certificate path.

\subsection{Route, review, prove}

The current operational formulation can be summarized in one line:

\begin{enumerate}
  \item \textbf{Route}: normalize the action and determine the posture it requires.
  \item \textbf{Review}: escalate when policy, ambiguity, or boundary semantics require human oversight.
  \item \textbf{Prove}: close the action into replayable evidence and exportable proof.
\end{enumerate}

\subsection{Externality-aware admissibility}

The reference implementation sharpens the original model in one important way: admissibility now depends not only on \emph{what} action is proposed, but also on \emph{what boundary it crosses}. Let $a$ denote a proposed governed action and let
\[
X(a) = \bigl(d, v, o, t, p, c, w, s, r, \epsilon\bigr)
\]
denote its externality context, where $d$ is destination type, $v$ is destination visibility, $o$ is destination ownership, $t$ is destination-tenant provenance, $p$ is account provenance, $c$ is client provenance, $w$ records whether the action is an external write, $s$ records sensitive-payload posture, $r$ captures reversibility or destructive posture, and $\epsilon$ captures approval enforceability.

The route step therefore evaluates a pair $(a, X(a))$ rather than an isolated action label. This makes the admission problem more honest for enterprise use. A low-level tool invocation and a high-level business side effect are not equivalent once public destinations, personal accounts, destructive writes, or observe-only runtimes are involved.

\subsection{Five-checkpoint contract}

We treat the five-checkpoint contract as the canonical operational form of PCAA.

\begin{table}[t]
\centering
\small
\caption{Five-checkpoint PCAA contract.}
\begin{tabularx}{\linewidth}{@{}>{\raggedright\arraybackslash}p{0.21\linewidth}>{\raggedright\arraybackslash}p{0.18\linewidth}X@{}}
\toprule
Checkpoint & SDK method & Purpose \\
\midrule
Pre-action admissibility & \texttt{guard} & Evaluate whether the action is allowed, warned, blocked, simulated first, or approval-gated before side effects occur. \\
Action open & \texttt{createAction} & Create the stable trust object that the rest of the certificate path attaches to. \\
Assumption capture & \texttt{recordAssumption} & Preserve what the runtime believed, expected, or depended on at decision time. \\
Approval checkpoint & \texttt{waitForApproval} & Pause or externally hold execution when policy requires human review. \\
Outcome closure & \texttt{updateOutcome} & Write final result, evidence, and closure status for replay and proof. \\
\bottomrule
\end{tabularx}
\end{table}

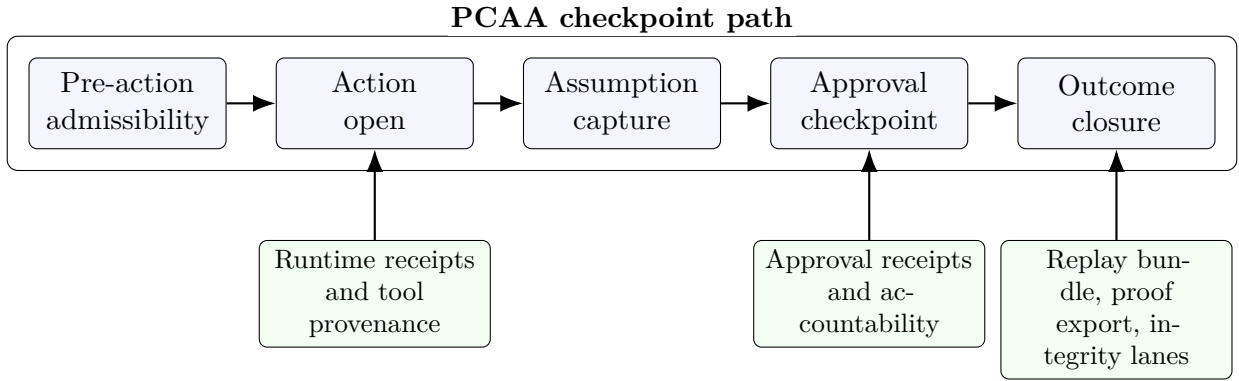
\begin{figure*}[t]
\centering
\begin{tikzpicture}[
  node distance=0.8cm and 0.65cm,
  box/.style={draw, rounded corners=3pt, align=center, minimum height=1.2cm, minimum width=2.35cm, text width=2.35cm, fill=blue!4},
  receipt/.style={box, fill=green!5, minimum width=2.8cm, text width=2.8cm, font=\small},
  lane/.style={draw, rounded corners=4pt, inner sep=8pt},
  arrow/.style={-{Latex[length=3mm]}, thick}
]
\node[box] (route) {Pre-action\\admissibility};
\node[box, right=of route] (open) {Action\\open};
\node[box, right=of open] (assume) {Assumption\\capture};
\node[box, right=of assume] (approve) {Approval\\checkpoint};
\node[box, right=of approve] (close) {Outcome\\closure};

\draw[arrow] (route) -- (open);
\draw[arrow] (open) -- (assume);
\draw[arrow] (assume) -- (approve);
\draw[arrow] (approve) -- (close);

\node[lane, fit=(route)(open)(assume)(approve)(close), label={[fill=white, inner sep=1pt]above:{\textbf{PCAA checkpoint path}}}] {};

\node[below=1.2cm of open, receipt] (runtime) {Runtime receipts\\and tool provenance};
\node[below=1.2cm of approve, receipt] (human) {Approval receipts\\and accountability};
\node[below=1.2cm of close, receipt] (proof) {Replay bundle, proof\\export, integrity lanes};

\draw[arrow] (runtime.north) -- (open.south);
\draw[arrow] (human.north) -- (approve.south);
\draw[arrow] (proof.north) -- (close.south);
\end{tikzpicture}
\caption{PCAA model: checkpoint flow plus attached receipts and proof closure.}
\label{fig:pcaa-overview}
\end{figure*}

\subsection{Core claim}

The central claim is therefore stronger than ``record runtime logs.'' High-value actions should not be admitted solely because the underlying platform allows them. They should be admitted only when a certificate path exists that can:

\begin{itemize}
  \item name the runtime boundary,
  \item preserve the governing route,
  \item preserve review semantics where required,
  \item disclose the externality boundary of the action itself,
  \item close into replayable and exportable proof,
  \item keep final authority explicit.
\end{itemize}

\section{Neutral Runtime Contract}

\subsection{Why neutrality matters}

Agent runtimes keep changing. Governance cannot restart every time the runtime changes. A natural response is to treat runtime integration as an adapter problem rather than a one-off product compatibility story. PCAA should show up as a portable execution sequence rather than a runtime-specific abstraction.

\begin{definition}[Neutral runtime contract]
For a runtime integration $r$, define the neutral contract
\[
\mathcal{N}(r) = \bigl(f_r, m_r, b_r, \alpha_r, s_r, \sigma_r\bigr),
\]
where $f_r$ is the runtime family, $m_r$ is the adapter mode, $b_r$ is the effective execution boundary, $\alpha_r$ is the declared authority split, $s_r$ is the session projection capability, and $\sigma_r$ is the receipt projection capability. Two runtimes are governance-comparable when they differ in vendor or product shape but can still be projected into the same neutral contract vocabulary.
\end{definition}

This definition is operational rather than rhetorical. The reference implementation exposes runtime families, adapter modes, authority split, adapter contract requirements, and checkpoint fields as explicit schema material instead of leaving them as prose-only documentation.

\subsection{Runtime families}

The neutral runtime contract groups runtimes into a small set of portable families.

\begin{table}[t]
\centering
\caption{Neutral runtime families used in the reference implementation.}
\begin{tabularx}{\linewidth}{@{}>{\raggedright\arraybackslash}p{0.2\linewidth}X@{}}
\toprule
Family & Operational interpretation \\
\midrule
OpenAI-compatible gateway & Request or tool-call governance at an API-shaped boundary. \\
Managed agent platform & Hosted agent systems that own sessions, tool execution, and state progression. \\
Framework SDK runtime & Code-first runtimes where governance can attach directly to tools, graph nodes, or workflow edges. \\
Tool-hook runtime & Developer-facing runtimes with commands, shell hooks, or pre-execution tool interception. \\
Observer or import-only runtime & Systems where evidence can be imported after execution but inline preemption is weak or absent. \\
\bottomrule
\end{tabularx}
\end{table}

\subsection{Adapter modes}

Runtime families alone are insufficient. The implementation also distinguishes how governance attaches.

\begin{table}[t]
\centering
\caption{Adapter modes and their disclosed interception posture.}
\begin{tabularx}{\linewidth}{@{}>{\raggedright\arraybackslash}p{0.18\linewidth}>{\raggedright\arraybackslash}p{0.16\linewidth}X@{}}
\toprule
Mode & Interception depth & Practical meaning \\
\midrule
Inline SDK & Strong & Governance code is embedded directly inside the runtime path and can participate in guard, action open, approval, and closure. \\
Lifecycle hook & Strong & The runtime emits pre/post execution hooks that let OSuite intervene close to tool use without replacing the runtime core. \\
Session adapter & Moderate & Structured session and tool events are forwarded through a local adapter with partial but meaningful control. \\
Gateway & Moderate & Governance attaches at an API or protocol boundary before the runtime proceeds. \\
Bridge & Selective & A managed platform forwards structured events and receipts, but some execution semantics remain owned by the remote runtime. \\
Observer & Weak & OSuite receives receipts or outcomes after execution and cannot claim deterministic inline interception. \\
\bottomrule
\end{tabularx}
\end{table}

This distinction matters because governance honesty depends on saying where interception is actually strong, moderate, selective, or weak.

\subsection{Boundary and coverage honesty}

The neutral contract is useful only if the execution boundary is named precisely. In the reference implementation, this is captured by the execution-boundary kind and the session-projection requirement. A runtime that emits only post-hoc outcomes is not equivalent to a runtime that can stop a tool call before side effects.

Let $\Gamma(r)$ denote the governance coverage posture of runtime $r$. The paper treats $\Gamma(r)$ as a disclosed qualitative property, not a hidden marketing abstraction. At minimum, it should encode:
\begin{itemize}
  \item whether pre-action admissibility is inline or post-hoc,
  \item whether approvals can actually pause execution,
  \item whether session and tool identities survive replay,
  \item whether runtime-stage receipts are complete, partial, or absent.
\end{itemize}

This is why the neutral contract uses both a family and a mode. A managed platform and a gateway may both look externally hosted, but their replay and intervention guarantees differ materially.

\subsection{Authority split}

The reference implementation makes a three-way authority split explicit:

\begin{enumerate}
  \item \textbf{Workspace authority}: OSuite remains authoritative for tenant scope, approvals, replay exports, and proof closure.
  \item \textbf{Runtime authority}: the runtime remains authoritative for actual tool execution, local mutation, and native trace generation.
  \item \textbf{External governor authority}: systems such as Trust Boundary may add pre-execution policy enforcement or runtime safety controls close to execution.
\end{enumerate}

This split is important because it prevents an all-or-nothing framing. A system may have strong runtime governance and still not replace workspace-visible certificate closure.

\begin{proposition}[Authority closure remains tenant-comparable]
Suppose two runtimes $r_1$ and $r_2$ project governed actions into the same checkpoint contract and preserve final authority as a PCAA-visible closure field. Then replay and buyer-visible review can remain tenant-comparable even when execution semantics differ between $r_1$ and $r_2$.
\end{proposition}

\noindent\textit{Informal justification.} The action certificate does not require identical execution engines. It requires that the decision route, approval state, outcome closure, and authority designation survive in a common structure. Runtime-specific receipts add detail, but they do not redefine the final control object.

\subsection{Approval enforceability classes}

The implementation adds a practical distinction that earlier formulations did not name sharply enough. Governance surfaces should say not only that an approval existed, but \emph{how enforceable} that approval was with respect to side effects. Let
\[
\epsilon(a) \in
\left\{
\begin{array}{l}
\texttt{pre\_execution\_gate}, \\
\texttt{observe\_only}, \\
\texttt{delegated\_runtime\_control}, \\
\texttt{runtime\_controlled}
\end{array}
\right\}
\]
be the approval-enforceability class for action $a$.

\begin{itemize}
  \item \texttt{pre\_execution\_gate}: OSuite or an attached governor can pause before side effects.
  \item \texttt{observe\_only}: evidence is imported after execution or after an agent abort, so OSuite cannot claim deterministic preemption.
  \item \texttt{delegated\_runtime\_control}: another runtime-layer governor enforced the decision, while PCAA preserved the tenant-visible certificate path.
  \item \texttt{runtime\_controlled}: execution completed under runtime control without a hard pre-execution review gate.
\end{itemize}

This classification is important because it lets replay and buyer review distinguish true fail-closed authority from post-hoc visibility.

\subsection{External governor projection}

The implementation reserves a first-class projection for Microsoft Agent Governance Toolkit (AGT). The important point is not brand inclusion. The important point is what AGT contributes:

\begin{itemize}
  \item runtime-stage receipts,
  \item approval receipts,
  \item release-compatibility surface,
  \item trust-boundary policy stages such as pre-input, pre-tool, post-tool, and pre-output.
\end{itemize}

PCAA absorbs those receipts without surrendering final authority. This is the clean integration line for managed enterprise runtimes: let the runtime or governor enforce close to execution, but let PCAA remain the tenant-visible route-review-prove layer.

\subsection{Schema surface}

The contract is not only conceptual. The reference implementation exposes a machine-readable runtime-governance interface that declares schema version, runtime families, adapter modes, authority split, action-envelope fields, checkpoint contract, plugin scope, and protocol lanes. That interface also accommodates action-boundary disclosure such as \texttt{externality\_context}. This matters because the operational specification can be checked against an executable surface rather than left as prose alone.

\section{Portable Action Envelope and Accountability}

\subsection{Action envelope as the portable trust object}

The reference implementation no longer treats the certificate only as an abstract tuple. It materializes it as a portable action envelope with a versioned schema. This envelope is the object that survives runtime churn.

At minimum, the envelope can declare:

\begin{itemize}
  \item action identity and declared goal,
  \item tenant and agent scope,
  \item runtime family and adapter mode,
  \item execution boundary kind,
  \item governance authority,
  \item authorization context,
  \item workflow context,
  \item evaluation context,
  \item externality context,
  \item accountability context,
  \item optional runtime identity, trace context, partner inputs, or protocol-lane declarations.
\end{itemize}

\subsection{Authorization, workflow, and evaluation context}

The implementation attaches more than action type and status. The authorization context names allowed systems, allowed actions, data domains, escalation mode, and whether human review is required. The workflow context preserves workflow identity, stage, business process, and connected systems. The evaluation context preserves feedback-loop and recommendation linkage.

These additions matter because they move PCAA away from being just a policy verdict store. The action becomes interpretable in the language an operator, buyer, or reviewer actually needs: what was attempted, under which workflow, against which systems, and whether evaluation or review was expected.

\subsection{Accountability context}

Named accountability is one of the clearest upgrades in the implementation studied here. The accountability chain includes:

\begin{itemize}
  \item requested-by identity,
  \item policy owner,
  \item runtime owner,
  \item final authority,
  \item approver lineage when approval receipts exist.
\end{itemize}

This is operationally significant. Many adjacent systems can produce traces, checks, or logs, but fewer preserve a tenant-readable answer to the question ``who was allowed to do what under which authority?''

\subsection{Externality context}

An especially important addition is the externality context. This object makes the certificate legible in the language enterprise reviewers actually use when assessing an agent action. The reference implementation can carry:

\begin{itemize}
  \item destination type and destination host,
  \item destination visibility and destination ownership,
  \item destination-tenant provenance,
  \item account provenance (\texttt{company}, \texttt{personal}, or \texttt{unknown}),
  \item client provenance,
  \item whether the action is an external write,
  \item whether sensitive payload is involved,
  \item reversibility and destructive posture,
  \item approval enforceability.
\end{itemize}

This matters because actor identity alone is not enough. In enterprise review, the dangerous question is often ``whose account, which client, and what external destination did the action actually use?''

\subsection{Runtime-stage receipts and approval receipts}

The implementation formalizes runtime-stage receipts over an ordered stage set:
\[
\texttt{pre\_input}, \quad \texttt{pre\_tool}, \quad \texttt{post\_tool}, \quad \texttt{pre\_output}.
\]

These receipts are intentionally orthogonal to PCAA closure. They preserve what an external runtime governor observed or enforced. Approval receipts do something analogous for human checkpoints: they preserve actor, decision, workflow identity, timing, and source.

The certificate path is stronger when these receipts are present, but the critical design rule is that transport or enforcement receipts do not replace final certificate closure. They feed the envelope and proof path; they do not redefine authority.

\subsection{Assumptions and decision basis}

The five-checkpoint model includes assumption capture explicitly. This is not cosmetic. A replayable action record should preserve what the runtime believed or relied on at decision time. Without that basis, later replay degenerates into a historical status report rather than a meaningful governance artifact.

\section{Proof, Replay, and Layered Integrity}

\subsection{Proof bundle closure}

The current proof model projects governed actions into replay and proof objects rather than leaving them as dashboard-only state. At the action level, the proof path can summarize:

\begin{itemize}
  \item decision outcome,
  \item policy and tool provenance,
  \item runtime profile,
  \item workload identity and trace context,
  \item authorization boundary,
  \item externality boundary,
  \item workflow accountability,
  \item evaluation closure,
  \item partner capability inputs,
  \item optional integrity or attestation state.
\end{itemize}

The important practical consequence is replayability. An operator or reviewer should be able to inspect what happened without depending on one vendor's live runtime session view.

\begin{definition}[Proof bundle projection]
For a governed action $a$, let
\[
\mathcal{P}(a) = \bigl(S(a), E(a), V(a), M(a)\bigr),
\]
where $S(a)$ is the human-readable summary, $E(a)$ is the evidence payload carried forward from the action certificate, $V(a)$ is the projected verification state, and $M(a)$ is the manifest material derived from canonicalized bundle content. In the current implementation, $M(a)$ is attached after canonicalization of subject, summary, evidence, and verification fields.
\end{definition}

In the reference implementation, the proof object is no longer only a conceptual endpoint. It is a structured bundle with schema version, generation timestamp, summary, evidence, projected verification checks, and manifest digest material.

\subsection{Proof graph intuition}

Earlier formulations framed a proof graph and commitment tuple more mathematically. In the reference implementation, that intuition is realized through structured summaries, receipts, manifests, and verification projections. The proof graph is therefore best understood as the set of preserved decision inputs and closure artifacts attached to the action:

\begin{itemize}
  \item route decision,
  \item runtime-stage receipts,
  \item approval receipts,
  \item assumptions and evidence references,
  \item accountability facts,
  \item manifest and verification status.
\end{itemize}

\subsection{Replay determinism and drift checks}

The replay layer in the reference implementation is stronger than a passive export. It recomputes action envelope, risk result, policy snapshot, and manifest payload, then checks for drift. Four failure classes are especially important:

\begin{itemize}
  \item missing evidence required for replay,
  \item manifest mismatch against the canonical replay payload,
  \item policy snapshot drift,
  \item materially drifted risk result.
\end{itemize}

This means replay is not merely ``load old JSON and display it.'' It is an active comparison between stored governance state and recomputed governance state under the current replay contract.

\subsection{Non-execution proof and boundary attestations}

The current implementation also sharpens what counts as successful proof closure. A governed action does not need to execute in order to produce a meaningful proof object. The proof path now distinguishes:

\begin{itemize}
  \item \texttt{execution\_proof},
  \item \texttt{pending\_execution\_proof},
  \item \texttt{partial\_execution\_proof},
  \item \texttt{non\_execution\_proof}.
\end{itemize}

This is practically important. Blocked actions, approval-paused actions, and agent-aborted actions should still close into replayable evidence rather than disappear as missing execution logs. The proof attestation set also now includes an explicit externality-boundary projection so the exported artifact can state destination visibility, account provenance, client provenance, and approval enforceability together.

\subsection{Receipt completeness}

Let $R_{\mathrm{exp}}(a)$ denote the expected runtime stages for action $a$ and let $R_{\mathrm{obs}}(a)$ denote the observed runtime-stage receipts after normalization. Receipt completeness can be summarized as
\[
\rho(a) =
\begin{cases}
\frac{|R_{\mathrm{exp}}(a) \cap R_{\mathrm{obs}}(a)|}{|R_{\mathrm{exp}}(a)|}, & |R_{\mathrm{exp}}(a)| > 0, \\
1, & |R_{\mathrm{exp}}(a)| = 0 \text{ and receipts are present}, \\
0, & \text{otherwise.}
\end{cases}
\]

This formulation matters because it preserves an honest distinction between complete, partial, and missing receipt coverage. In heterogeneous runtimes, receipt completeness is not a cosmetic metric. It is part of what determines how much trust the replay path can claim.

\begin{table}[t]
\centering
\caption{Proof-path checks that support replay and verification.}
\begin{tabularx}{\linewidth}{@{}>{\raggedright\arraybackslash}p{0.22\linewidth}X@{}}
\toprule
Check & Operational consequence \\
\midrule
Manifest digest present & Allows later fingerprint comparison against a canonicalized proof payload. \\
Projected verification attached & Preserves a verifier-facing summary without exposing all internal payloads. \\
Approval summary complete & Makes review semantics portable across runtimes and workflows. \\
Externality boundary declared & Makes destination, account, client, and enforceability posture portable in the proof export. \\
Receipt coverage computed & Discloses whether runtime enforcement evidence is complete, partial, or absent. \\
Outcome closure recorded & Ensures the certificate closes to a stable terminal state rather than a hanging action log. \\
\bottomrule
\end{tabularx}
\end{table}

\subsection{Dual-lane and future-lane integrity}

The layered integrity architecture remains useful, and the current implementation clarifies its role:

\begin{enumerate}
  \item \textbf{Fast commitment lane}: lightweight commitments or attestations that make silent drift harder.
  \item \textbf{Portable trust lane}: typed trust materials, attestations, or credentials that can travel across systems.
  \item \textbf{Future verification lane}: stronger receipts or verifiable computation artifacts that may mature later.
\end{enumerate}

This separation is important for technical honesty. PCAA does not need every action to be publicly or cryptographically verified to remain meaningful. The stronger near-term claim is that certificate closure, replay, and integrity references are preserved in a stable shape that can later project into stronger verification ecosystems.

\subsection{Why chain integration is not ornamental}

Chain-backed or receipt-bearing integrity is only useful if it strengthens portability, challengeability, or anti-drift guarantees. The current design explicitly rejects the idea that a chain receipt, transport receipt, or wallet attestation could replace PCAA closure. Such artifacts are adjacent trust materials, not the authority layer itself.

\subsection{Buyer-visible consequence}

This proof model matters commercially as well as technically. Buyers and security reviewers do not only ask whether an action was blocked. They ask whether approvals can be shown, whether receipts survive export, whether runtime coverage is disclosed honestly, and whether evidence can be replayed after the runtime stack changes. The proof path is what lets PCAA answer those questions without collapsing back into product marketing.

\section{Implementation Snapshot}

\subsection{Reference executable surface}

The reference implementation exposes the PCAA loop as a concrete SDK and API surface. The core runtime flow is:

\begin{enumerate}
  \item \texttt{guard(context)}
  \item \texttt{createAction(action)}
  \item \texttt{recordAssumption(assumption)}
  \item \texttt{waitForApproval(id)} when required
  \item \texttt{updateOutcome(id, outcome)}
\end{enumerate}

In practice, this means the certificate path is no longer merely conceptual. There is a concrete object lifecycle that runtimes and operators can invoke.

\subsection{Operational artifact surface}

The implementation also exposes adjacent artifact surfaces that matter for the operational interpretation of the model:

\begin{itemize}
  \item governance proof snapshots and proof bundles,
  \item action trace and replay state,
  \item PCAA health summaries,
  \item runtime inventory and adapter catalogs,
  \item buyer assurance and control-matrix outputs,
  \item connector-posture and enterprise-readiness packets.
\end{itemize}

These surfaces matter because they show that PCAA is now split across at least three layers:

\begin{enumerate}
  \item the action-level contract,
  \item the runtime-boundary contract,
  \item the buyer-visible assurance and export layer.
\end{enumerate}

\begin{table}[t]
\centering
\caption{Executable surfaces that anchor the implementation snapshot.}
\begin{tabularx}{\linewidth}{@{}>{\raggedright\arraybackslash}p{0.22\linewidth}X@{}}
\toprule
Surface & What it contributes \\
\midrule
SDK loop & The concrete checkpoint lifecycle used by connected runtimes and sample integrations. \\
Runtime-governance contract API & A machine-readable schema for runtime families, adapter modes, checkpoints, action-envelope fields, plugin scope, and protocol lanes. \\
Replay and proof bundle services & Exportable proof objects, manifest digests, verification projections, and drift-aware replay. \\
Connect guide templates & Minimal smoke, support triage, and release-gate examples that show how the abstract contract is expected to be used in practice. \\
Activation and readiness summaries & A buyer-ready champion path that begins with connection and ends with proof and verification. \\
\bottomrule
\end{tabularx}
\end{table}

\subsection{Coverage disclosure}

One of the strongest implementation choices is explicit coverage disclosure. The system distinguishes inline SDK, hook, gateway, bridge, session, and observer modes instead of pretending all runtimes provide equal governance depth. This is operationally preferable to treating imported evidence as indistinguishable from deterministic inline interception.

\subsection{Boundary-aware enterprise deployment}

Recent implementation work does not replace the governance core. Instead, it sharpens the operational vocabulary around the existing core. In particular, the system now exposes direct executable anchors for:

\begin{itemize}
  \item externality-aware action envelopes,
  \item approval-enforceability summaries,
  \item connector posture as a first-class assurance packet,
  \item company-versus-personal account provenance,
  \item public-facing readiness and disclosure surfaces.
\end{itemize}

This matters because the enterprise-facing interpretation of PCAA is no longer implicit in dashboard prose alone. It is expressed through stable surfaces that can be checked against the underlying governance contract.

\subsection{Reference integration patterns}

The implementation also illustrates concrete integration patterns rather than only hypothetical deployments. Representative examples include a smoke-test agent, a support-triage agent, and a release-gate agent. These examples are valuable because they demonstrate three distinct uses of the checkpoint contract:

\begin{itemize}
  \item low-risk connection validation,
  \item workflow routing with assumptions and metadata,
  \item approval-gated execution for production-side changes.
\end{itemize}

This spread is analytically useful because it shows that the same checkpoint contract can support low-risk validation, workflow routing, and approval-gated production actions.

\subsection{Partner and extension boundary}

The implementation also makes partner and extension boundaries explicit. AGT, AgentGuard, AIM, Attestix, Microsoft enterprise posture inputs, and optional Web3 lanes can contribute trust or evidence, but they do so under a strict rule: PCAA remains final authority. This keeps composability from degenerating into an incoherent control model.

\subsection{Champion scenario}

The reference implementation is also unusually explicit about what a successful activation looks like. A single ``buyer-ready governed proof'' scenario is modeled as a five-step path: create workspace authority, issue runtime credentials, submit a first governed action, export a proof bundle, and verify a fingerprint. This is useful because it anchors abstract governance claims in an observable readiness sequence.

\subsection{Current limitations}

Several limitations remain visible from the implementation and should remain visible in the paper:

\begin{itemize}
  \item runtime depth is heterogeneous and sometimes only partially enforced,
  \item stronger verifier-network or public-scale cryptographic verification is still future-facing,
  \item benchmark-quality public evaluation for the full certificate path is still limited,
  \item some empirical claims should be treated as historical baseline results rather than statements about the current system state.
\end{itemize}

For this reason, the paper should be read as an implementation-informed systems account, not as a claim that every quantitative result remains invariant under later system revisions.

\section{Public Validation Protocol}

\subsection{Protected corpus, public aggregates}

The evaluation in this paper is public in method but not in raw benchmark contents. We report aggregate metrics derived from a protected executable corpus and withhold the full scenario list, exact policy thresholds, scalar-heuristic weights, and connector-specific routing rules. The reason is practical rather than cosmetic: the manuscript is intended to support public scrutiny of the control method without turning the benchmark into an implementation playbook.

The protected corpus expands 24 executable seed templates into 96 traces by replaying the same governed action families across four runtime families: framework SDK, OpenAI-compatible gateway, managed agent platform, and observer or import-only runtime. This aggregate profile is large enough to show route quality, review burden, replay behavior, and ablation sensitivity, even though it does not expose each underlying scenario.

\begin{table}[!htbp]
\centering
\caption{Disclosure-safe corpus profile used for public validation.}
{\small
\begin{tabularx}{\linewidth}{@{}>{\raggedright\arraybackslash}p{0.42\linewidth}>{\centering\arraybackslash}p{0.18\linewidth}X@{}}
\toprule
Dimension & Count & Public interpretation \\
\midrule
Seed templates & 24 & Protected internal benchmark seeds spanning four control outcomes \\
Executable traces & 96 & Seed templates replayed across four runtime families \\
Runtime families & 4 & Framework SDK, gateway, managed platform, and observer or import-only modes \\
Decision buckets & 24 / 20 / 28 / 24 & Aggregate counts for allow, simulate-first, require-approval, and block \\
Boundary classes & 6 & Internal read, simulation-guarded mutation, trust-degraded high impact, external egress, prompt abuse, destructive denied path \\
\bottomrule
\end{tabularx}
}
\end{table}

The protocol is designed to answer a limited but important public question: can a runtime-neutral governance layer preserve routing and proof behavior across heterogeneous runtimes, and can the role of key PCAA components be made visible through ablation?

\subsection{Route quality}

Table~\ref{tab:pcaa-public-route-quality} compares PCAA with two weaker baselines: static rules and a scalar heuristic. On this benchmark, PCAA attains perfect route fidelity, while the baselines lose substantial route resolution. The more important observation is comparative rather than absolute. The weaker methods either over-block or collapse distinct cases into a thresholded verdict, whereas the certificate-bearing path preserves the benchmark's decision taxonomy.

That result should be interpreted carefully. A perfect score on a protected benchmark does not mean that runtime governance is a solved problem, nor does it imply uniform performance under future deployment drift. It means that, for this benchmark family, the implemented route semantics align closely with the decision categories the benchmark is designed to exercise.

\begin{table}[!htbp]
\centering
\caption{Aggregate route quality on the protected public-validation corpus.}
\label{tab:pcaa-public-route-quality}
{\small
\begin{tabularx}{\linewidth}{@{}>{\raggedright\arraybackslash}p{0.28\linewidth}*{4}{>{\centering\arraybackslash}p{0.15\linewidth}}@{}}
\toprule
Method & Exact accuracy & Macro-F1 & Severe recall & Block precision \\
\midrule
Static rules & 0.583 & 0.493 & 1.000 & 0.400 \\
Scalar heuristic & 0.375 & 0.296 & 0.778 & 0.182 \\
PCAA runtime & \textbf{1.000} & \textbf{1.000} & \textbf{1.000} & \textbf{1.000} \\
\bottomrule
\end{tabularx}
}
\end{table}

\begin{figure}[!htbp]
\centering
\begin{tikzpicture}
\begin{axis}[
  pcaa chart,
  title={Aggregate Route Quality on Protected Public Validation},
  ylabel={Score},
  ymin=0,
  ymax=1.05,
  ybar,
  bar width=10pt,
  symbolic x coords={Exact Accuracy,Macro-F1,Severe Recall,Block Precision},
  xtick=data,
  enlarge x limits=0.15,
  legend columns=3,
  legend style={at={(0.5,1.12)}, anchor=south},
]
\addplot[fill=pcaastatic, draw=pcaastatic] coordinates {
  (Exact Accuracy,0.583)
  (Macro-F1,0.493)
  (Severe Recall,1.000)
  (Block Precision,0.400)
};
\addplot[fill=pcaascalar, draw=pcaascalar] coordinates {
  (Exact Accuracy,0.375)
  (Macro-F1,0.296)
  (Severe Recall,0.778)
  (Block Precision,0.182)
};
\addplot[fill=pcaaruntime, draw=pcaaruntime] coordinates {
  (Exact Accuracy,1.000)
  (Macro-F1,1.000)
  (Severe Recall,1.000)
  (Block Precision,1.000)
};
\legend{Static Rules,Scalar Heuristic,PCAA Runtime}
\end{axis}
\end{tikzpicture}
\caption{Aggregate route-quality comparison on the protected validation corpus.}
\end{figure}
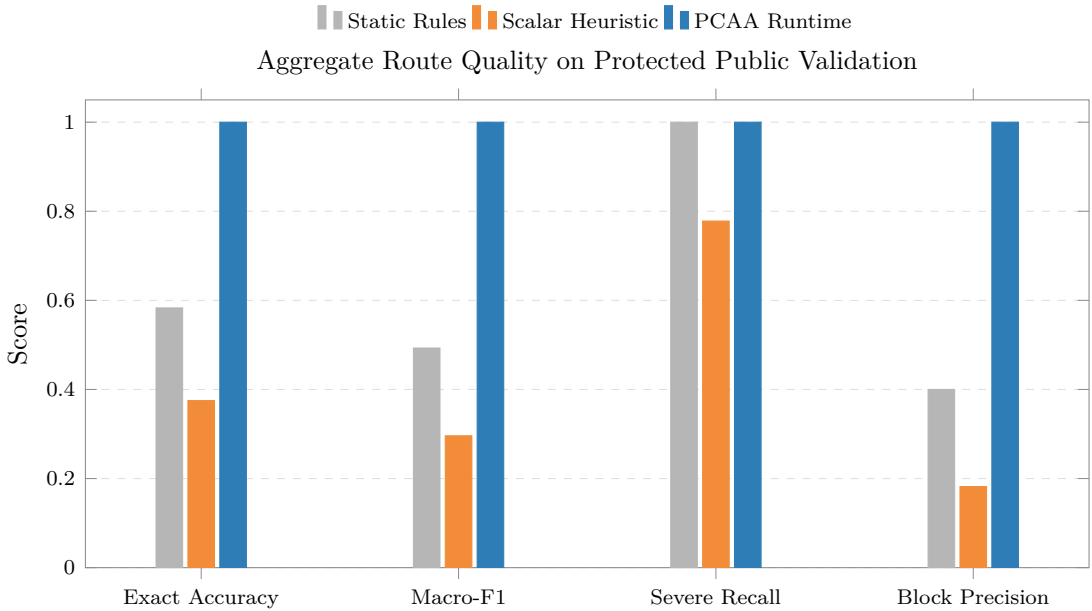

\subsection{Review burden and proof closure}

High route quality alone would be insufficient. A governance layer is useful only if it allocates scarce review attention in a plausible way and still closes actions into replayable evidence. Table~\ref{tab:pcaa-public-burden-proof} therefore reports both review burden and proof-path behavior.

PCAA routes 29.2\% of traces into explicit review and 20.8\% into simulate-first handling, while keeping hard blocks to 25.0\%. This is a visibly different operating posture from the baselines, which depend more heavily on hard blocks and expose less structured review behavior. On the proof side, the manifest is stable and replay-readiness is complete for the evaluated traces. Receipt completeness is lower, at 0.516, because the corpus intentionally mixes strong inline and gateway controls with weaker observer or import-only runtimes. That lower value is not an accident; it reflects heterogeneous runtime depth and is one of the practical constraints the paper is trying to describe honestly.

\begin{table}[!htbp]
\centering
\caption{Review burden and proof-path behavior for the PCAA runtime.}
\label{tab:pcaa-public-burden-proof}
{\scriptsize
\begin{tabularx}{\linewidth}{@{}>{\raggedright\arraybackslash}p{0.40\linewidth}>{\centering\arraybackslash}p{0.18\linewidth}>{\centering\arraybackslash}p{0.18\linewidth}>{\centering\arraybackslash}p{0.18\linewidth}@{}}
\toprule
Metric family & Value 1 & Value 2 & Value 3 \\
\midrule
Review burden & review 0.292 & simulate 0.208 & block 0.250 \\
Proof and replay & manifest 1.000 & replay 1.000 & receipts 0.516 \\
Latency envelope & mean 0.35 ms & P95 0.54 ms & P99 0.72 ms \\
\bottomrule
\end{tabularx}
}
\end{table}

\subsection{Ablation}

The most useful evidence in favor of PCAA is not the headline score alone, but the pattern of degradation when key components are removed. Three ablations are especially important:
\begin{enumerate}
  \item removing externality context,
  \item removing approval-enforceability handling,
  \item removing the integrity lane.
\end{enumerate}

Removing externality context lowers exact routing from 1.000 to 0.875 and reduces explicit review routing for actions whose risk comes from boundary crossing rather than from raw action verbs alone. Removing approval-enforceability handling lowers exact routing to 0.917 and shifts workload toward a less disciplined review posture. Removing the integrity lane does not change route quality, but it collapses manifest stability to 0 and lowers proof closure to 0.708. The three ablations fail differently, which is precisely the point: routing, review, and proof are coupled, but they are not interchangeable.

\begin{table}[!htbp]
\centering
\caption{Ablation results on the protected validation corpus.}
{\scriptsize
\begin{tabularx}{\linewidth}{@{}>{\raggedright\arraybackslash}p{0.28\linewidth}>{\centering\arraybackslash}p{0.16\linewidth}>{\centering\arraybackslash}p{0.14\linewidth}>{\centering\arraybackslash}p{0.18\linewidth}>{\centering\arraybackslash}p{0.14\linewidth}@{}}
\toprule
Ablation & Exact accuracy & Review rate & Manifest stability & Proof closure \\
\midrule
Full PCAA runtime & \textbf{1.000} & 0.292 & \textbf{1.000} & \textbf{1.000} \\
No externality & 0.875 & 0.167 & 1.000 & 1.000 \\
No enforceability & 0.917 & 0.375 & 1.000 & 1.000 \\
No integrity lane & 1.000 & 0.292 & 0.000 & 0.708 \\
\bottomrule
\end{tabularx}
}
\end{table}

\begin{figure}[!htbp]
\centering
\begin{tikzpicture}
\begin{axis}[
  pcaa chart,
  title={PCAA Ablation Study},
  ylabel={Score},
  ymin=0,
  ymax=1.05,
  ybar,
  bar width=12pt,
  symbolic x coords={Full,No Externality,No Enforceability,No Integrity},
  xtick=data,
  xticklabel style={rotate=12, anchor=east},
  enlarge x limits=0.18,
  legend columns=3,
  legend style={at={(0.5,1.12)}, anchor=south},
]
\addplot[fill=pcaaruntime, draw=pcaaruntime] coordinates {
  (Full,1.000)
  (No Externality,0.875)
  (No Enforceability,0.917)
  (No Integrity,1.000)
};
\addplot[fill=pcaawarn, draw=pcaawarn] coordinates {
  (Full,0.292)
  (No Externality,0.167)
  (No Enforceability,0.375)
  (No Integrity,0.292)
};
\addplot[fill=pcaarisk, draw=pcaarisk] coordinates {
  (Full,1.000)
  (No Externality,1.000)
  (No Enforceability,1.000)
  (No Integrity,0.000)
};
\legend{Exact accuracy,Review rate,Manifest stability}
\end{axis}
\end{tikzpicture}
\caption{Ablation sensitivity of the public validation protocol.}
\end{figure}
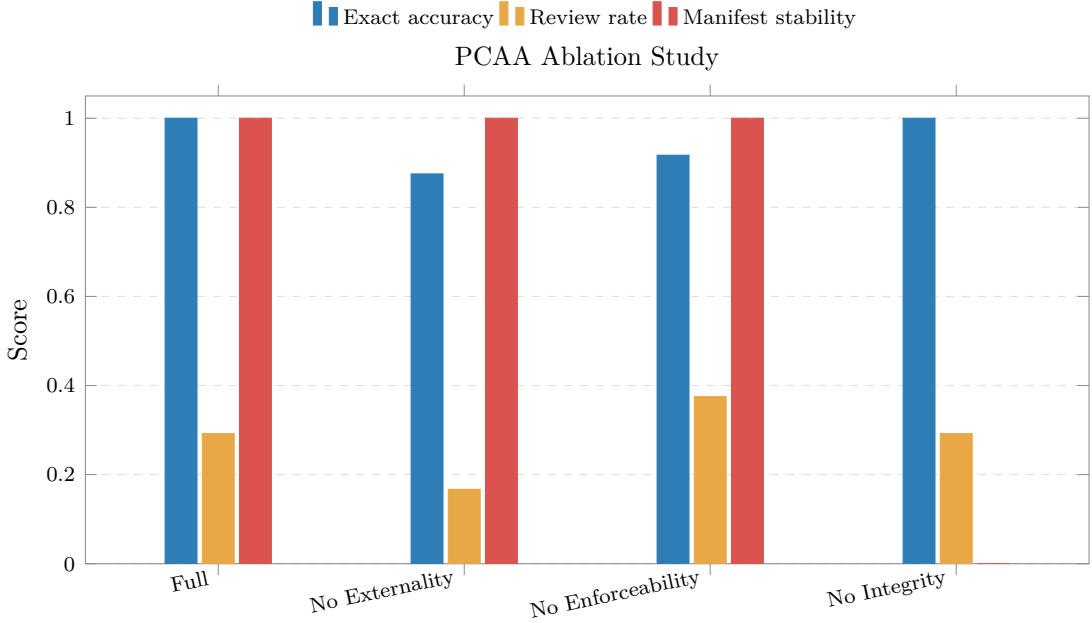

\subsection{Disclosure boundary}

The disclosure boundary is asymmetric by design. Public readers can inspect the metric definitions, aggregate corpus composition, baseline comparisons, proof-path outcomes, and ablation sensitivity. They cannot inspect the raw scenario corpus, exact thresholds, exact heuristic weights, or connector-specific routing rules.

This choice narrows what the paper can claim. It is appropriate for a public research manuscript and for buyer-facing scrutiny, but it is not the same as releasing a full academic benchmark package. The evidentiary basis for the reported results is still concrete: the corpus is executable, the aggregate metrics are machine-derived, and the ablations identify which parts of PCAA are carrying the observed behavior. What remains withheld is implementation-sensitive detail rather than the existence of a runnable evaluation.

\section{Discussion}

\subsection{Where the novelty claim actually lies}

The most defensible novelty claim in PCAA is architectural. The paper does not introduce model evaluation, runtime tracing, programmable rails, or cryptographic attestation in isolation. Its contribution is to recombine these adjacent capabilities around a different primary object: an action certificate created at decision time and closed only after review and outcome have been recorded.

That shift has several consequences. First, heterogeneous runtime normalization becomes part of the core problem rather than an integration detail. Second, selective review semantics become part of the control method itself rather than a workflow added later. Third, externality is treated as part of admissibility, not merely as post-hoc metadata. Fourth, replay and proof are attached to the same object that carried the original decision. Finally, integrity references can be layered onto the certificate path without replacing certificate closure as the tenant-visible authority record.

\subsection{Relation to adjacent systems}

This framing helps explain what PCAA is, but also what it is not. It is not merely a policy table: a rule engine may emit allow or deny, yet still lose the review path, accountability chain, and replay basis that make a decision meaningful under later scrutiny. It is also not simply an LLM judge. Model-based assessment may contribute to routing, but the certificate path is designed to remain inspectable and replayable even when the underlying model changes or when a deployment avoids heavy inference in the control loop.

The same distinction appears in runtime governance. Enforcement close to execution is often necessary, and systems such as AGT or Trust Boundary can provide valuable receipts and pre-tool controls. But runtime enforcement alone does not solve tenant-visible authority. If those receipts cannot be projected back into a certificate with explicit approval and closure semantics, the governance story fragments across tools and dashboards. PCAA is intended to hold those pieces together without pretending that every runtime exposes the same depth of control.

The enterprise-oriented refinements described in this version of the paper should be read in that light. Boundary facts such as destination visibility, personal-versus-company accounts, approved-client posture, and approval enforceability may sound operational rather than theoretical, but they sharpen the central research question. A runtime-neutral certificate is not useful in practice if it abstracts away exactly the boundary facts that determine whether an action is acceptable.

\subsection{Limits of the current evidence}

The paper also has clear limits. Runtime depth is heterogeneous, so not every deployment can claim the same pre-execution control. The empirical section is based on an executable but protected benchmark corpus rather than on a fully public benchmark release. Some adjacent product references will evolve over time, and the implementation-grounded vocabulary used here may shift faster than conventional academic terminology.

These constraints do not nullify the core claim, but they do narrow it. The paper should be read as an implementation-informed systems proposal with bounded public validation, not as a claim of universal third-party verification or invariant empirical performance across all future runtime revisions. The strongest reading of the evidence is therefore comparative: PCAA appears useful to the extent that route, review, and proof can remain attached to a common certificate path under runtime churn, and less useful to the extent that a deployment cannot expose meaningful receipts or enforce approvals before side effects.

\section{Conclusion}

This paper has argued for a simple but consequential shift in how agent governance is modeled. When high-value actions are routed through heterogeneous runtimes, the most durable trust object is neither the model report nor the hosted session record, but the action certificate that captures the decision, the review semantics, and the evidence needed for replay.

In the reference implementation studied here, that certificate path is made concrete through five checkpoints, a portable action envelope, a neutral runtime contract, and a proof model that keeps runtime receipts, approval receipts, and integrity projections on one path without denying heterogeneous runtime depth. The empirical section does not show that governance is solved. It does show that, on a protected but executable benchmark, the route-review-prove structure degrades in informative ways when externality, approval enforceability, or integrity are removed.

If agent runtimes continue to proliferate, the practical challenge is unlikely to be choosing one permanent execution substrate. It is more likely to be preserving a stable governance object while execution substrates keep changing. PCAA is offered as one approach to that problem.

\appendix
\section{Reference Implementation Notes}

This appendix summarizes the operational surfaces that matter most for interpreting the reference implementation discussed in the paper.

\begin{table}[h]
\centering
\small
\caption{Primary implementation surfaces relevant to the paper.}
\begin{tabularx}{\linewidth}{@{}>{\raggedright\arraybackslash}p{0.32\linewidth}X@{}}
\toprule
Implementation surface & Primary paper relevance \\
\midrule
Checkpoint lifecycle interface & Five checkpoints, admission semantics, and outcome closure \\
Portable action-envelope schema & Accountability, authorization, workflow, evaluation, and externality contexts \\
Runtime and approval receipt pipeline & Runtime-stage normalization and human-review evidence \\
Proof and replay bundle pipeline & Replay-oriented summaries, manifest digests, and verification projections \\
Authority and accountability projection & Final-authority designation and approver lineage \\
Neutral runtime contract surface & Runtime families, adapter modes, and authority split \\
External-governor compatibility layer & Projection of runtime-adjacent enforcement signals such as AGT stages \\
Assurance and export surfaces & Buyer-visible governance framing and deployment-facing summaries \\
\bottomrule
\end{tabularx}
\end{table}

\section{Evaluation Provenance}

The quantitative results in Section~9 combine two artifact classes:
\begin{enumerate}
  \item a protected executable benchmark corpus maintained with the reference implementation;
  \item a disclosure-safe public evaluation surface that reports aggregate metrics, ablations, and proof-path outcomes without exposing implementation-sensitive scenario details.
\end{enumerate}

The public manuscript intentionally withholds the raw scenario corpus, exact policy thresholds, exact scalar-heuristic weights, and connector-specific decision rules. Instead, it reports protected-corpus composition at the aggregate level, together with route-quality metrics, review-burden metrics, proof and replay metrics, and ablation sensitivity.

This separation is intentional. The goal of the public paper is to make the technical claim inspectable without turning the benchmark into a complete implementation playbook. Aggregate evidence remains reproducible inside the maintained reference implementation, while the public manuscript remains suitable for broad circulation, buyer review, and preprint distribution.

\bibliographystyle{plainnat}
\bibliography{references}

\end{document}